\documentclass[12pt]{amsart}
\usepackage{amsthm}
\usepackage{amsmath}
\usepackage{amssymb}
\usepackage[latin1]{inputenc}
\usepackage[T1]{fontenc}
\usepackage{graphicx}
\newtheorem{definition}{Definition}

\newtheorem{remark}{Remark}
\newtheorem{theorem}{Theorem}

\newtheorem{proposition}{Proposition}

\newcommand{\vp}{\varphi}

\newcommand{\C}{\mathbb{C}}
\newcommand{\N}{\mathbb{N}}
\newcommand{\R}{\mathbb{R}}

\newcommand{\Z}{\mathbb{Z}}
\newcommand{\be}{\begin{equation}}
\newcommand{\eeq}{\end{equation}}
\newcommand{\bet}{\begin{equation*}}
\newcommand{\eeqt}{\end{equation*}}
\newcommand{\bea}{\begin{eqnarray}}
\newcommand{\eeqa}{\end{eqnarray}}
\newcommand{\beat}{\begin{eqnarray*}}
\newcommand{\eeqat}{\end{eqnarray*}}

\newcommand{\h}[1]{\mathcal{#1}}
\newcommand{\hil}{\mathcal{H}}

\newcommand{\br}{\mathcal{B}(\R)}
\newcommand{\brr}{\mathcal{B}(\R^2)}
\newcommand{\lh}{\mathcal{L(H)}}

\newcommand{\sfq}{\mathsf{Q}}%position spectral measure
\newcommand{\sfp}{\mathsf{P}}%momentum spectral measure
\newcommand{\p}{\mathsf{p}}
\newcommand{\E}{\mathsf{E}}
\newcommand{\F}{\mathsf{F}}
\newcommand{\G}{\mathsf{G}}

\newcommand{\tr}[1]{\mathrm{tr}\left[ {#1} \right]}

\begin{document}
\title{A note on the Pauli problem in light of approximate joint measurements}

\author{Jussi Schultz}
\address{Turku Centre for Quantum Physics, Department of Physics and Astronomy, University of Turku, FI-20014 Turku, Finland}
\email{jussi.schultz@utu.fi}

\begin{abstract}
We show that there exist informationally incomplete phase space observables such that the Cartesian margins are informationally equivalent with position and momentum. This shows that it is possible to reconstruct the position and momentum distributions of a quantum system from the statistics of a single observable, and thus a single measurement, even though the state of the system is not uniquely determined by the statistics.
\end{abstract}

\maketitle

\section{Introduction}
The Pauli problem is the classic example concerning the possibility of performing quantum state tomography. In a footnote in \cite{Pauli}, Pauli noted that the problem whether or not the state of a quantum system is uniquely determined by the position and momentum distributions ``has still not been investigated in all its generality''. In modern terminology, this is a question about the informational completeness \cite{Prugovecki1977} of  the pair $(\sfq,\sfp)$ of position and momentum observables. It is now known that the answer to this question is in the negative: $(\sfq, \sfp)$ is informationally incomplete. As a matter of fact, a wide variety  of counterexamples have been constructed showing that different states can have the same position and momentum distributions \cite{Prugovecki1977, Reichenbach, Corbett1978, Pavicic1986}. This has been viewed  as ``an illustration of the {\em surplus information}  \cite{Weizsacker1987} coded in a quantum (pure) state when compared with its classical counterpart'' \cite{Busch1989}.

%%%%%%%%%%

The purpose of this paper is to address the Pauli problem in a slightly stronger form. We pose the following question:
\begin{itemize}
\item[(Q)]  If the position and momentum distributions of a quantum system can be reconstructed from the statistics of a single observable, is the observable necessarily informationally complete?
\end{itemize}
\noindent
We show that the answer to (Q) is also in the negative. We construct explicitly a phase space observable whose Cartesian margins are informationally equivalent \cite{Davies1970} with position and momentum, thus allowing the reconstruction of the corresponding distributions from the marginal statistics, even though the observable is informationally incomplete. This result further illustrates the difference in the role of position and momentum in determining the state of a system in quantum mechanics,  as opposed to that in  classical mechanics. Indeed, the state of a classical system is given by its generalized position and momentum coordinates which can, in principle, be determined by a single measurement, whereas in quantum mechanics even the simultaneous determination (in the  sense of (Q)) of position and momentum does not guarantee unique state determination. 

%%%%%%%%%%

The paper is organized as follows. In Section \ref{Preliminaries} we lay out the general framework for this study. We review the relevant definitions and results concerning informational equivalence, informational completeness and phase space observables. In Section \ref{QP-tomography} we give the counterexample to (Q). We also prove the existence of informationally complete phase space observables whose margins do not suffice to determine position and momentum. In Section \ref{Reconstruction} we consider the process of reconstructing position and momentum distributions from the statistics of a single phase space observable. We do this using two different methods: the Fourier theory and the method of moments.  The conclusions are given in Section \ref{conclusions}.

%%%%%%%%%%

\section{Preliminaries}\label{Preliminaries}
Let $\hil= L^2 (\R)$ be the Hilbert space associated with a quantum system such as a spinless particle confined to move in one dimension or a single mode electromagnetic field. Let $\{ h_n \vert n=0,1,2,\ldots \}$ be the orthonormal basis of $\hil$ consisting of the Hermite functions and denote $ \h D=\textrm{lin}\{ h_n \vert n=0,1,2,\ldots \} $ so that $\overline{\h D}=\hil$. 

%%%%%%%%%%

The states of the system are represented by positive trace class operators $\rho$ with unit trace. The observables are represented by normalized positive operator measures $\E: \h B(\R^n) \rightarrow \lh$ where $\h B(\R^n)$ denotes the $\sigma$-algebra of Borel subsets of $\R^n$, and $\lh$ is the set of bounded operators on $\hil$. An observable is called sharp if it is projection valued, that is, $\E(X)^2 =\E(X)$ for all $X\in \h B(\R^n)$. For a system in a state $\rho$ the measurement outcome statistics of an observable $\E$ is given by the probability measure $\p^\E_\rho:\h B(\R^n)\rightarrow [0,1]$, $\p^\E_\rho (X) =\textrm{tr}[\rho \E(X)]$. 
\begin{definition}
Let $\E:\h B (\R^n)\rightarrow \lh$ and $\F:\h B (\R^m)\rightarrow \lh$ be observables. 
\begin{itemize} 
\item[(a)] If for any two states $\rho$ and $\sigma$, $\p^\E_\rho =\p^\E_\sigma$ implies $\p^\F_\rho =\p^\F_\sigma$, then the state distinction power of $\E$ is greater than or equal to that of $\F$. 
\item[(b)] If for any two states $\rho$ and $\sigma$, $\p^\E_\rho =\p^\E_\sigma$ if and only if $\p^\F_\rho =\p^\F_\sigma$, then $\E$ and $\F$ are informationally equivalent. 
\item[(c)] If for any two states $\rho$ and $\sigma$, $\p^\E_\rho =\p^\E_\sigma$ implies $\rho =\sigma$, then $\E$ is informationally complete.
\end{itemize}
\end{definition}

%%%%%%%%%%

Let $Q$ and $P$ denote the selfadjoint position and momentum operators and let $\sfq,\sfp:\br\rightarrow \lh$ be the corresponding sharp observables. Define the Weyl operators $W(q,p)= e^{i\frac{qp}{2}} e^{-iqP} e^{ipQ}$, $(q,p)\in\R^2$, and for each positive trace class operator $T$ with unit trace define the phase space observable $\G^T:\brr\rightarrow \lh$ by 
\begin{equation*}
\G^T (Z) =\frac{1}{2\pi} \int_Z W(q,p) TW(q,p)^*\, \textrm{d}q\, \textrm{d}p
\end{equation*}
for all $Z\in\brr$. The operator $T$ is called the generating operator of $\G^T$. The Cartesian margins of $\G^T$ are the  unsharp position and momentum observables $\mu^T*\sfq, \nu^T*\sfp:\br\rightarrow \lh $ defined as 
\begin{align}
\left( \mu^T*\sfq \right) (X) &=\int_X \mu^T(X-x)\, \textrm{d}\sfq(x),\nonumber\\
\left( \nu^T*\sfp \right) (Y) &=\int_Y \nu^T (Y-y)\, \textrm{d}\sfp(y),\nonumber
\end{align}
for all $X,Y\in\br$. The convolving measures are determined by the generating operator $T$ so that $\mu^T (X) =\textrm{tr}[T\sfq (-X)]$ and $\nu^T (Y) =\textrm{tr}[T\sfp (-Y)]$.

%%%%%%%%%%%

In general the state distinction power of, say, the unsharp position $\mu^T *\sfq$ does not exceed that of $\sfq$. However, they may be   informationally equivalent. This is the case if and only if the support of the Fourier transform of the convolving measure is $\R$ \cite{Heinonen2004}, that is, $\textrm{supp}\, \widehat{\mu}^T =\R$ where
\begin{equation}\label{FourierQ}
\widehat{\mu}^T (p) = \frac{1}{\sqrt{2\pi}}\int e^{-ipx}\, \textrm{d}\mu^T(x) =\frac{1}{\sqrt{2\pi}}\textrm{tr}[T W(0,p)]
\end{equation}
for all $p\in\R$. The same is of course true for the unsharp momentum $\nu^T*\sfp$, where the Fourier transform is now given by 
\begin{equation}\label{FourierP}
\widehat{\nu}^T (q) = \frac{1}{\sqrt{2\pi}}\textrm{tr}[T W(-q,0)].
\end{equation}

%%%%%%%%%%

The informational completeness of the phase space observable $\G^T$ can also be characterized in terms of the  Weyl transform $(q,p)\mapsto \tr{TW(q,p)}$ of the generating operator. Indeed, a sufficient condition for informational completeness has been known since \cite{Ali1977}. In the recent paper \cite{Kiukas2012}, this  question was exhaustively resolved and  several equivalent necessary and sufficient conditions were obtained. For the purpose of the present paper the following characterization is convenient (see \cite[Prop. 4]{Kiukas2012}).
\begin{theorem}\label{infotheorem}
The phase space observable  $\G^T$ is informationally complete if and only if the support of $(q,p)\mapsto \tr{TW(q,p)}$ is $\R^2$. 
\end{theorem}
\noindent
An important class of informationally complete phase space observables are those whose generating operator satisfies $\tr{TW(q,p)}\neq 0 $ for all $(q,p)\in\R^2$. Such operators are called regular \cite{Werner1984}. Clearly any $T$ determined by a Gaussian wavefunction  is regular.

%%%%%%%%%%

\section{Main results}\label{QP-tomography}
The phase space observables $\G^T$ are archetypes of approximate joint observables for position and momentum. It is even possible to choose  $\G^T$  in such a way that that the Cartesian margins $\mu^T*\sfq$ and $\nu^T*\sfp$ are informationally equivalent with $\sfq$ and $\sfp $. Indeed, consider the simplest case $T=\vert h_0\rangle\langle h_0\vert$ so that for any state the corresponding  probability density is the Husimi $Q$-function  \cite{Husimi1940} of the state. In this case the convolving measures are given by $\mu^T(X) =\nu^T(X) =\frac{1}{\sqrt{\pi}}\int_X e^{-x^2}\, dx$ so that $\widehat{\mu}^T(p)=\widehat{\nu}^T(p)=\frac{1}{\sqrt{2\pi}} e^{-\frac{p^2}{4}}$ which confirms informational equivalence. However, this particular  observable is also informationally complete, as can be seen from $\langle h_0 \vert W(q,p) h_0\rangle = e^{-\frac{1}{4}(q^2+p^2)}$. The following Proposition shows that this is not in general the case.

%%%%%%%%%%

\begin{proposition}\label{counterexample1}
There exist informationally incomplete phase space observables whose margins are informationally equivalent with position and momentum.
\end{proposition}
\begin{proof}
Let $\vp=\chi_{\scriptscriptstyle{[-1/2,1/2]}}$,  the characteristic function of  the interval $\left[-\tfrac{1}{2}, \tfrac{1}{2}\right]$, and define 
$$
T=\frac{1}{2}\vert \vp\rangle\langle \vp\vert + \frac{1}{2} \vert\widehat{\vp}\rangle \langle \widehat{\vp}\vert
$$
where $\widehat{\vp}$ denotes the Fourier transform of $\vp$, that is,
$$
\widehat{\vp} (p) = \frac{1}{\sqrt{2\pi}}\int e^{-ipx} \vp(x)\, \textrm{d}x =\frac{1}{\sqrt{2\pi}} \frac{\sin (p/2)}{p/2}
$$
for all $p\in \R\setminus \{0\}$ and $\widehat{\vp}(0) =\frac{1}{\sqrt{2\pi}}$. Now consider the observable  $\G^T$ and, in particular, the margins $\mu^T*\sfq$ and $\nu^T*\sfp$. Using the fact that $\vp$ and $\widehat{\vp}$ are even functions we obtain the convolving measures
\begin{align}
\mu^T(X) &=\nu^T (X) = \textrm{tr}[T\sfq(X)] \nonumber\\
&= \int_X \frac{1}{2} \left( \vert\vp (x)\vert^2 + \vert \widehat{\vp}(x)\vert^2  \right)\, \textrm{d}x\nonumber
\end{align}
for all $X\in\br$. Equations \eqref{FourierQ} and \eqref{FourierP} now give the Fourier transforms as (see Figure \ref{Fourierfigure}) 
\begin{align*}
\widehat{\mu}^T(p) =&\, \widehat{\nu}^T (p) =  \frac{1}{2\sqrt{2\pi}} \int e^{ipx} \vert \vp(x)\vert^2\, \textrm{d}x \\
&+  \frac{1}{2\sqrt{2\pi}} \int \overline{\vp(x)}\vp(x+p)\, \textrm{d}x \nonumber\\
=& \left\{ \begin{array}{ll}
\tfrac{1}{2\sqrt{2\pi}}\left(1-\vert p\vert + \frac{\sin (p/2)}{p/2}\right), & \textrm{ when }\vert p\vert \leq 1 \\
\tfrac{1}{2\sqrt{2\pi}}  \frac{\sin(p/2)}{p/2}, & \textrm{ otherwise}.
\end{array}\right.\nonumber
\end{align*}
This shows that $\widehat{\mu}^T(p) =0$ if and only if $p=2n\pi$, $n\in\Z\setminus\{ 0\}$, that is, $\textrm{supp}\, \widehat{\mu}^T = \textrm{supp}\, \widehat{\nu}^T =\R$. 

%%%%%%%%%%

\begin{figure}[h]
\includegraphics[width=0.8\columnwidth]{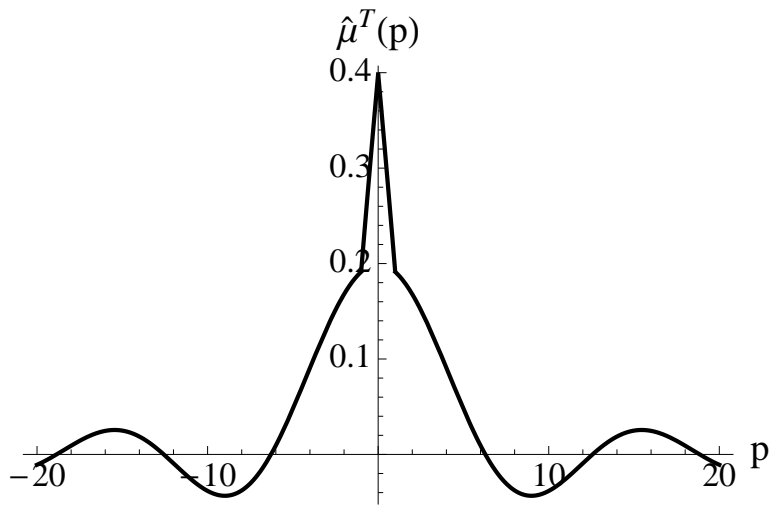}
\caption{The plot of $\widehat{\mu}^T$ as a function of $p$.}\label{Fourierfigure}
 \end{figure}

%%%%%%%%%%

Now consider the informational completeness of $\G^T$. Using the fact that $F^{-1}W(q,p)F =W(-p,q)$ for all $(q,p)\in\R^2$, where $F$ is the Fourier-Plancherel operator, we get
\begin{align}
\textrm{tr}[TW(q,p)] &= \frac{1}{2} e^{-i\frac{qp}{2}}\int e^{ipx} \overline{\vp(x)} \vp(x-q)\, \textrm{d}x \nonumber\\
 &+ \frac{1}{2} e^{i\frac{qp}{2}}\int e^{iqx} \overline{\vp(x)} \vp(x+p)\, \textrm{d}x.\nonumber
\end{align}
The first integral is zero when $\vert q\vert \geq 1$ and the second one vanishes for $\vert p\vert \geq 1$. Thus $\textrm{tr}[TW(q,p)]=0$ whenever $\vert q\vert, \vert p\vert\geq 1$ which shows that the support of $(q,p)\mapsto \textrm{tr}[TW(q,p)]$ is not the whole $\R^2$. According to Theorem \ref{infotheorem} this means that $\G^T$ is not informationally complete.
\end{proof}

%%%%%%%%%%

Note that  since the generating operator $T$  is typically linked to the state of the probe system used in the measurement of $\G^T$, the generating operator of Proposition \ref{counterexample1} is also physically meaningful. Indeed, in the case that $T$ is (mathematically) the state of the probe, this can be realized by  randomly  preparing the probe system in each experimental run into the state $\vert \varphi \rangle\langle \varphi\vert$ or  $\vert \widehat{\varphi} \rangle\langle \widehat{\varphi}\vert$. That is, in each run the probe state is localized either in position or momentum space with probability $1/2$. 

%%%%%%%%%%%

We close this section by showing that the informational completeness of a phase space observable does not imply that the margins are informationally equivalent with position and momentum. In this case it seems difficult to give a physically relevant counterexample, but mathematically it can be constructed.
 
%%%%%%%%%%%

\begin{proposition}\label{counterexample2}
There exist informationally complete phase space observables whose margins are not informationally equivalent with position and momentum.
\end{proposition}
\begin{proof}
To begin with, we pick any regular generating operator $T_0$, i.e., one that  satisfies $\tr{T_0W(q,p)} \neq 0$ for all $(q,p)\in\R^2$. Then for any positive $f\in L^1 (\R^2)$ such that $\int f(q,p)\, \textrm{d} q\textrm{d} p =1$, the trace class operator
$$
f*T_0 = \int f(q,p) W(q,p) T_0 W(q,p)^*\, \textrm{d} q\textrm{d} p
$$
is positive and of unit trace \cite{Werner1984}. Furthermore, we have $\tr{f*T_0 W(q,p)}= 2\pi\, \widehat{f}(-p,q)\tr{T_0 W(q,p)}$ so that in view of Equations \eqref{FourierQ} and \eqref{FourierP}, and Theorem \ref{infotheorem} we need to find the function $f$ in such a way that $\textrm{supp} \, \widehat{f} = \R^2$ but the supports of the functions $p\mapsto \widehat{f} (-p,0)$ and  $q\mapsto \widehat{f} (0,-q)$ are proper subsets of  $\R$.

%%%%%%%%%%%

To begin the construction, we define for any $r>0$ the function 
$$
\widehat{g}_{0,r} (q,p) =\frac{1}{2\pi r}e^{-q^2} \left(\chi_{\left[-\frac{r}{2}, \frac{r}{2}\right] } * \chi_{\left[-\frac{r}{2}, \frac{r}{2}\right]}\right) (p)
$$ 
so that 
$$
g_{0,r} (x,y) =\frac{1}{2r\pi^{3/2}} \frac{1-\cos (ry)}{y^2} e^{-\frac{x^2}{4}}.
$$
In particular, $g_{0,r}\in L^1 (\R^2)$, $g_{0,r} \geq 0$, $\int g_{0,r} (x,y) \, \mathrm{d}x\mathrm{d}y=1$ and $\widehat{g}_{0,r}(q,p)=0 $ if and only if $\vert p\vert \geq r$. Then for any $\theta \in[0,\frac{\pi}{2}]$ define 
$$
g_{\theta,r} (x,y) = g_{0,r}(x\cos\theta + y\sin\theta, -x\sin\theta + y\cos\theta) ,
$$
so that $\widehat{g}_{\theta,r}$ is nonzero only on the strip which is in an angle of $\theta$ with respect to the $q$-axis (see Figure \ref{zeros}). In particular, $\widehat{g}_{\theta,r} (q,0) = 0 $ when  $\vert q\vert \geq \frac{r}{\sin\theta}$ and $\widehat{g}_{\theta,r} (0,p)  = 0 $ when $\vert p\vert \geq \frac{r}{\cos\theta}$.

%%%%%%%%%%

\begin{figure}[h]
\begin{center}
\includegraphics[width=0.7\columnwidth]{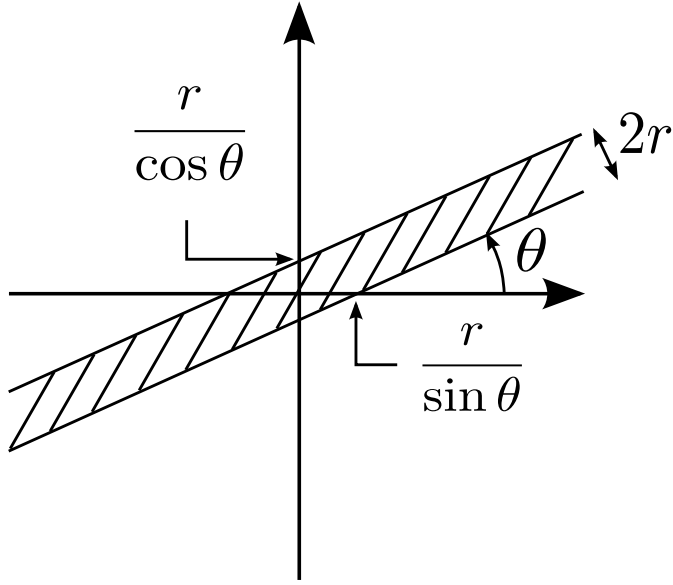}
\caption{The shaded region represents the strip in phase space where the function $\widehat{g}_{\theta,r}$ is nonzero.}\label{zeros}
\end{center}
 \end{figure}

%%%%%%%%%%

 We then define $\theta_n = \frac{\pi}{2^{n+1}}  $ and $r_n = \frac{\sin \theta_n}{2^n-1}$ and the function $g$ via the series 
$$
g(x,y) = \sum_{n=1}^\infty \sum_{k=1}^{2^n-1} \frac{1}{2^{n+k}} g_{ k\theta_n ,r_n}(q,p)
$$
so that $g\in L^1 (\R^2)$ and $\widehat{g}(x,y)=0$ if and only if $\widehat{g}_{k\theta_n,r_n}(x,y)=0$ for all $n\in\N$ and  $k=1,\ldots, 2^{n}-1$. Finally, we define  $f(x,y) = C \left(g(x,y) + g(-x,y) \right)$, where $C=\frac{1}{2}\left( \int g(x,y)\,  \mathrm{d}x \mathrm{d}y \right)^{-1} $ so that $f\in L^1 (\R^2)$, $f(x,y)\geq 0$ for all $(x,y)\in\R^2$, and $\int f(x,y) \, \mathrm{d}x\, \mathrm{d}y = 1$.

%%%%%%%%%%

In order to see that this function has the desired properties, let $B$ be a ball with radius $\epsilon >0$ centered at $(q_0,p_0)\in\R^2$. Due to the symmetry properties of $f$ we may without loss of generality assume that $q_0,p_0\geq 0$. Since $\{  k\theta_n \mid   n\in\N, k=1,\ldots, 2^{n}-1 \}$ is dense in $[0,\frac{\pi}{2}]$, there exists an $n_0\in\N$ and a $ k_0\in \{ 1, \ldots, 2^{n_0}-1\}$ such that the line $p = q\tan(k_0\theta_{n_0})$  intersects with $B$. In particular, $(r_0\cos(k_0\theta_{n_0}), r_0 \sin(k_0\theta_{n_0}))\in B$ for some $r_0>0$ and since $\widehat{f} (r_0\cos(k_0\theta_{n_0}), r_0 \sin(k_0\theta_{n_0}))\neq 0$, we conclude that $\textrm{supp}\, \widehat{f} = \R^2$. However, $\widehat{f}(q,0) = 0 $ when 
$$
\vert q\vert \geq \sup_{n,k} \tfrac{r_n}{\sin{k\theta_n}} = \sup_{n,k} \left( \tfrac{1}{2^n-1} \tfrac{\sin\theta_n}{\sin (k\theta_n)} \right)= \frac{1}{2}
$$ 
and $\widehat{f}(0,p) =0$ when 
$$
\vert p\vert \geq \sup_{n,k} \tfrac{r_n}{\cos{k\theta_n}} = \sup_{n,k} \left( \tfrac{1}{2^n-1} \tfrac{\sin\theta_n}{\cos (k\theta_n)} \right)= \frac{1}{2}, 
$$
so that the support of neither of these ``marginal'' functions is $\R$. Hence,  the observable $\G^{f*T_0}$ is informationally complete but the marginal observables are not informationally equivalent with position and momentum.
\end{proof}

%%%%%%%%%%

\begin{remark}
\rm
Obviously the counterexample for Proposition \ref{counterexample1} could also have been given via a similar construction. Indeed, if $T_0$ is as before and  we define $f_0 =C_0\left( g_{0,1} + g_{\pi/2,1}\right)$ where $C_0>0$ is the normalization constant, we find that $\G^{f_0*T_0}$ is informationally incomplete but the margins satisfy the appropriate conditions for informational equivalence. However, we found it useful to present the given counterexample since it is physically more relevant.
\end{remark}

%%%%%%%%%%

\section{Reconstructing position and momentum}\label{Reconstruction}
We will next demonstrate two methods for determining the position and momentum distributions from the marginal statistics of a phase space observable. The first one uses the theory of Fourier transforms and second one uses the statistical method of moments. For the first part only the minimal requirement of informational equivalence for the margins is required as for the second part we need a stronger condition, namely, that of exponential boundedness.  

%%%%%%%%%%

\subsection{Fourier theory}
Let $\G^T$ be such that $\mu^T*\sfq$ and $\nu^T*\sfp$ are informationally equivalent with the corresponding sharp observables $\sfq$ and $\sfp$, that is, $\textrm{supp}\, \widehat{\mu}^T=\textrm{supp}\, \widehat{\nu}^T =\R$. Let $\rho$ be an arbitrary state and consider the first marginal observable. By taking a Fourier transform of the probability measure we obtain $(\widehat{\mu^T*\p^\sfq_\rho})(p) =\sqrt{2\pi}\widehat{\mu}^T(p) \widehat{\p}^\sfq_\rho(p)$ for all $p\in\R$. It follows that for all $p\in\R$ such that $\widehat{\mu}^T(p)\neq 0$ we have 
\begin{equation}\label{FT_of_convolution}
\widehat{\p}^\sfq_\rho (p)=\frac{1}{\sqrt{2\pi}} \frac{(\widehat{\mu^T*\p^\sfq_\rho})(p)}{\widehat{\mu}^T(p)}.
\end{equation} 
Since we know that the Fourier transform $\widehat{\p}^\sfq_\rho $ is a bounded continuous function, we can take appropriate limits on the right-hand side of Eq. \eqref{FT_of_convolution} to determine $\widehat{\p}^\sfq_\rho $ for those $p\in\R$ for which $\widehat{\mu}^T(p)=0$.  Since the Fourier transform is injective, this uniquely determines the measure $\p^\sfq_\rho$.

%%%%%%%%%%

If we want to invert the Fourier transform to get explicitly the form of the position distribution, we need to assume that the right-hand side of Eq. \eqref{FT_of_convolution} is integrable. To that end, suppose that $\rho$ is a pure state given by a unit vector  $\psi\in\h D$. The position distribution $x\mapsto \vert \psi(x)\vert^2$ then belongs to the Schwartz space $\h S(\R)$, and hence also the Fourier transform $  \widehat{\p}^\sfq_\rho  $ is in $\h S(\R)$. Furthermore, since $\h S(\R)\subset L^1(\R)$ we may use the Fourier inversion theorem to obtain the position distribution
$$
\vert \psi(x)\vert^2 = \frac{1}{2\pi}\int e^{ixp} \frac{(\widehat{\mu^T*\p^\sfq_\psi})(p)}{\widehat{\mu}^T(p)}\, \textrm{d}p
$$
for almost all $x\in\R$. A similar treatment shows that we can reconstruct the momentum distribution from $\nu^T*\p^\sfp_\psi$.  The position and momentum distributions can therefore be reconstructed explicitly at least for the dense set of vector states $\h D$. Furthermore, this  method works also for any finite mixture $\rho =\sum_{n=0}^k c_n \vert \psi_n\rangle\langle \psi_n \vert$ where $(\psi_n)_{n=0}^k\subset \h D$ and $(c_n)_{n=0}^k\subset [0,1]$, $\sum_{n=0}^k c_n=1$. In other words, for any state $\rho$ whose matrix representation with respect to the basis $\{ h_n \vert n=0,1,2,\ldots\}$ is finite. It is known that  such states are dense in the set of all states (this follows, e.g., from \cite[Theorem 1]{Grumm1973}).

%%%%%%%%%%

\subsection{Method of moments}\label{method_of_moments}
The statistical method of moments  was presented in this context in \cite{Busch2008} and was further illustrated in \cite{Kiukas2009}. The idea is to reconstruct the moments of the position and momentum distributions from the moments of the marginal statistics.  In this case the informational equivalence of the observables is not sufficient to ensure the existence of the moments. Indeed, in  the counterexample of Proposition \ref{counterexample1}, the convolving measures $\mu^T$ and $\nu^T$ do not have any finite moments. Moreover, even the existence of finite moments does not guarantee that the probability measure is uniquely determined by the moments. In this sense the ability to reconstruct the moments of the position and momentum distributions does not necessarily mean that the actual distributions can be reconstructed. To circumvent this problem, we need to make a stronger assumption of exponential boundedness for the measures. 

%%%%%%%%%%

Recall that a probability measure $\mu :\h B(\R)\rightarrow [0,1]$ is exponentially bounded if there exists an $a>0$ such that 
$$
\int e^{a\vert x \vert} \, \textrm{d}\mu(x) <\infty.
$$
According to \cite[Prop. 2]{Lahti2005} (which is based on the proof of \cite[Prop. 1.5]{Simon1998}) a probability measure is exponentially bounded if and only if there exist positive constants $C,R>0$ such that  
\begin{equation}\label{moment_inequality}
 \big\vert\,  \mu[k]\,  \big\vert \leq C R^k k!
\end{equation}
for all $k\in\N$. In addition to the existence of all moments, exponentially bounded measures have the important property of being determinate. That is, if $\mu$ is exponentially bounded and $\nu:\br\to[0,1]$ is another probability measure such that $\nu[k] = \mu[k]$ for all $k\in\N$, then $\nu = \mu$. It is also worth noting that any exponentially bounded measure satisfies $\textrm{supp}\, \widehat{\mu}=\R$. This is due to the fact that the Fourier transform has an analytic continuation to the strip $\{ z\in\C\mid \vert\textrm{Im}(z)\vert < a\}$, and is therefore zero in at most countably many points. In particular, if $\mu^T$ and $\nu^T$ are exponentially bounded, then the corresponding observables $\mu^T*\sfq$ and $\nu^T*\sfp$ are informationally equivalent with $\sfq$ and $\sfp$. 

%%%%%%%%%%

Let again $\rho$ be a state and suppose that $\G^T$ is such that the marginal measures $\mu^T*\p^\sfq_\rho $ and $\nu^T*\p^\sfp_\rho$ are exponentially bounded. Note that since $\textrm{supp}\, \widehat{\mu^T*\p^{\sfq}_\rho} =\textrm{supp}\, \left(\sqrt{2\pi}\, \widehat{\mu}^T \,\widehat{\p}^{\sfq}_\rho \right)=\R$ and similarly for the second margin, we have $\textrm{supp}\, \widehat{\mu}^T =\textrm{supp}\, \widehat{\nu}^T = \R$ so that informational equivalence is guaranteed.  In particular, if the moments of the marginal distributions satisfy  \eqref{moment_inequality}, then we know that the corresponding observables are informationally equivalent with position and momentum.   Thus, the moment inequality \eqref{moment_inequality} may be viewed as an operational sufficient condition for informational equivalence.

%%%%%%%%%%

Consider again the first marginal probability measure $\mu^T*\p^\sfq_\rho$. We can now calculate the $k$th moment as
$$
\left(\mu^T*\p^\sfq_\rho \right) [k] =\sum_{n=0}^k \binom{k}{n}\mu^T[k-n] \p^\sfq_\rho[n].
$$
From this expression the moments of the position distribution can be solved recursively giving
$$
\p^\sfq_\rho [k] = \left(\mu^T*\p^\sfq_\rho \right)[k] - \sum_{n=0}^{k-1} \binom{k}{n} \mu^T[k-n] \p^\sfq_\rho [n].
$$
In other words, we are able to express the moments of the position distribution in terms of the operationally meaningful moments of the marginal statistics. If we know {\em a priori} that  $\p^\sfq_\rho$ is also exponentially bounded so that it is determined by its moments, then we have determined uniquely the position distribution. This is the case, for instance, when $\rho$  is a finite mixture of vector states from $\h D$, which is a dense set. A similar treatment can of course be carried out for the second margin. 

%%%%%%%%%%

It is worth noting that if the measures $\mu^T$ and $\nu^T$ are exponentially bounded, then $\mu^T*\p^\sfq_\rho$ and $\nu^T*\p^\sfp_\rho$ are exponentially bounded whenever $\p^\sfq_\rho$ and $\p^\sfp_\rho$ are such. Therefore, the exponential boundedness of  $\mu^T$ and $\nu^T$ guarantees that this   method can be used for a dense set of states. 

%%%%%%%%%%

\section{Conclusions}\label{conclusions}
We have shown that there is no direct connection between the informational completeness of a  phase space observable and the state distinction properties of its Cartesian margins. More precisely, we have shown that it is not possible to infer informational completeness from the condition that the margins are informationally equivalent with sharp position and momentum observables. This shows that it is possible to determine the position and momentum distributions from the statistics of a single measurement even though the state is not uniquely determined. We have also demonstrated the converse fact that informational completeness does not guarantee the possibility of reconstructing  position and momentum from the mere marginal statistics, but one may occasionally need take the detour via full state reconstruction.

\ 

\noindent
\textbf{Acknowledgments.} The author is grateful to Pekka Lahti for useful discussions and comments on the manuscript. The author was supported by the Finnish Cultural Foundation and the Academy of Finland Grant No. 138135.

%%%%%%%%%%

\end{document}